\documentclass[11pt,twoside,a4paper]{article}

\usepackage[english]{babel}
\usepackage[latin1]{inputenc}
\usepackage[T1]{fontenc} %Richtige Trennung
\usepackage{dsfont}
\usepackage{amssymb,amsmath,amsthm,amscd,mathrsfs,bbm}
\usepackage[left = 1.3in, right = 1.25in]{geometry}

\usepackage[font=small,labelfont=small]{caption}
\usepackage{fancyhdr}
\usepackage{titling}
\usepackage[titletoc]{appendix}
\usepackage{hyperref}
\usepackage{verbatim}
\usepackage{float}
\usepackage{caption}
\usepackage{subcaption}
\usepackage{tikz}
\usetikzlibrary{patterns}
\usetikzlibrary{shapes.geometric}
\usetikzlibrary{decorations.shapes}
\usetikzlibrary{decorations.pathreplacing}
\usepackage{cancel}
\usepackage{xcolor}

\newcommand{\neu}[1]{{\textcolor{black}{#1}}}

\theoremstyle{plain} \numberwithin{equation}{section}
\newtheorem{thm}{Theorem}[section]
\newtheorem{lemma}[thm]{Lemma}

\theoremstyle{definition}
\newtheorem{defn}[thm]{Definition}
\theoremstyle{remark}

\theoremstyle{remark}

\DeclareMathOperator{\Var}{Var}
\DeclareMathOperator{\Cov}{Cov}
\DeclareMathOperator{\argmin}{argmin}
\DeclareMathOperator{\vol}{vol}
\DeclareMathOperator{\median}{median}

\newcommand{\R}{{\mathbb R}}

\newcommand{\Z}{{\mathbb Z}}
\newcommand{\E}{{\mathbb E}}
\newcommand{\pr}{{\mathbb P}}

\newcommand{\e}{\mathcal{E}}
\newcommand{\B}{\mathcal{B}}

\title{\bf Entropy-based inhomogeneity detection in porous media}
\author{Patricia Alonso Ruiz and  Evgeny Spodarev}
\date{}
\begin{document}
\maketitle
\vspace*{-1.5cm}
\begin{abstract}
We study a change-point problem for random fields based on a univariate detection of outliers via the $3\sigma$-rule in order to recognize inhomogeneities in porous media. In particular, we focus on fibre reinforced polymers modeled by stochastic fibre processes with high fibre intensity and search for abrupt changes in the direction of the fibres. As a measure of change, the entropy of the directional distribution is locally estimated within a window that scans the region to be analyzed.
\end{abstract}

\section{Introduction}
Lightweight materials are highly demanded in many industrial applications, for instance in automobile, aerospace or wind turbine construction. Fibre reinforced polymers (FRP) constitute an important class of such materials, whose macroscopic properties are directly influenced by their microstructure, and in particular by the orientation of the fibres. With the aim of gaining a better understanding of this relation, methods of non-destructive characterization such as the analysis of $3$-dimensional images by micro-computed tomography ($\mu$CT), combined with stochastic microstructure modeling, are currently being investigated~\cite{OS09,R+12,RSVW14}. 
 
In the compression moulding process of a FRP, the fibres order themselves inside the raw material as a result of mechanical pressure. During this process, deviations from the requested direction may occur, creating undesirable fibre clusters and/or deformations. This kind of inhomogeneities are characterized by abrupt changes in the direction of the fibres and its detection is subject of study in the change-point analysis. In particular in this paper, we are interested in developing a method that is sensitive to changes and at the same time is considerably fast. In order to reduce the complexity of the problem carried out by spatial data, we investigate the random field generated by the entropy of the directional distribution of the fibres estimated in a moving scanning window that runs over the observed piece of material. In this way, the initial spatial change-point problem becomes a univariate detection of outliers on a random field.% The aim of this paper is to present an entropy-based method of detection of regions where such inhomogeneities appear. %As explained later in more detail, the inhomogeneity region will be assumed to be relatively small with respect to the whole scanned area. 

\begin{figure}[H]
\centering
\includegraphics[scale=0.9]{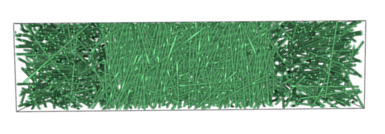}
\caption{A simulated fibre structure (courtesy of Katja Schladitz, ITWM Kaiserslautern).}
\end{figure}

%More precisely, the main idea consists in searching for outliers within the local entropy of the directional distribution of the fibres estimated on the moving scanning window. 
The entropy is estimated in a non-parametric manner by means of the entropy estimator introduced in~\cite{ARS17}. Due to its asymptotic normality, see~\eqref{eq CLT}, the $3\sigma$-rule yields a criterion to identify outliers. Searching for points above or below the given threshold depends on the specific application and both procedures are completely analogous. Moreover, this procedure may be applicable not only to the direction but to any other characteristic of interest, for instance the fibre length or the fibre curvature.

It is important to remark at this point that the use of the $3\sigma$-rule is justified by the asymptotic normality of the entropy estimator. As a result, the proposed method is reliable on structures with high fibre intensity. In the simulation study presented in Section~\ref{section simulations} we have therefore generated structures with an intensity of 5 fibres per unit volume. If we consider for instance a cube of side length $10$ $\mu$m filled with fibres of radius around $1.5$ $\mu$m and length $10$ $\mu m$, this corresponds to a fibre volume ratio of 35\%, a usual percentage in high-densely FRPs.

We have worked deliberately with simulated data for several reasons: on the one hand, the inhomogeneity region is known a priory and it can thus easily be compared with the output. This ensures an accurately testing of the performance of the detection algorithm. On the other hand, the segmentation of $3$-dimensional images of high-intensity fibre structures is still difficult and expensive.

In connection with this kind of detection problems one can find in the literature parametric scan-statistics~\cite{SJT12}, multivariate cumulative sums (cusum) methods~\cite{LZJ13} and change-point analysis~\cite{MJ93} for spatial data.

The fibre structures modeled in the present paper correspond to those having relatively short and rigid fibres that will be assumed to have the same length. More precisely, we consider the \textit{Boolean model}
\[
\Phi_\lambda=\bigcup_{Y_i\in\Pi_\lambda}(F_{Y_i}+Y_i),
\]
where $\Pi_\lambda=\{Y_i\}_{i\geq 1}\subseteq\R^3$ is a \textit{homogeneous Poisson point process} of intensity $\lambda>0$ and $F_{Y_i}$ denotes an independent copy of a line segment $F$ with a given length $\ell >0$. Each point $Y_i$ corresponds to the center of a fibre $F_{Y_i}$ and thus $\Phi_\lambda$ is a random set that models a system of independent fibres. To each center $Y_i$ a mark $\xi_i$ is attached that represents the (random) unit direction vector of the fibre $F_{Y_i}$. These marks are assumed to be i.i.d. random variables independent of the location of the points $Y_i$, and their distribution is supposed to have a continuous density $f$ with respect to a spherical surface measure. 

Based on the marked Poisson point process (MPPP) $\Psi_\lambda:=\{(Y_i,\xi_i)\}_{i\geq 1}$, a non-parametric estimator of the entropy of the directional distribution has been introduced in~\cite{ARS17}. Although the Poisson assumption is too strong to actually be considered in applications, we have focused on reproducing the property of high fibre volume ratio. Applying our methods to more realistic models such as hard-core processes is still a matter of future research.

The paper is organized as follows: in Section~\ref{section theory}, definitions and theoretic results concerning the estimation of entropy for the MPPP $\Psi_\lambda$ are presented. Section~\ref{section method} describes the proposed method of detection of inhomogeneities in the direction of the fibres and investigates the optimal size relationship between the inhomogeneity region and the scanning window size. Section~\ref{section simulations} is devoted to a simulation study of the estimators and the detection algorithm. The paper finishes in Section~\ref{section discussion} with a discussion of possible limitations of the method which are object of further research.

\section{Non-parametric entropy estimation.}\label{section theory}
The key element in the proposed method to detect inhomogeneities in a FPR is the use of the \textit{entropy} of the directional distribution of the fibres composing it. The concept of entropy was introduced by Shannon in~\cite{Sha48} to provide a way of measuring the closeness of a distribution to the uniform distribution. Given a random variable $\xi$ that takes values on an abstract measurable space $(M,\nu)$ and has distribution density $f_\xi\colon M\to\R$ with respect to the measure $\nu$, its entropy is defined as
\[
\e_{f_\xi}:=-\E[\log f_\xi(\xi)]=-\int_M\log f_\xi(x)\,f_\xi(x)\,\nu(dx).
\]
As mentioned in the introduction, the fibre structures being considered here are modeled by the MPPP
\begin{equation}\label{def MPP}
\Psi_\lambda=\{(Y_i,\xi_i)\}_{i\geq 1}\subseteq \R^3\times S^{2},\qquad\lambda>0,
\end{equation}
where $S^2$ denotes the unit sphere in $\R^3$. The Euclidean space will be equipped with the Lebesgue measure $\vol(dy):=dy$, and $S^2$ with the surface-area measure $\omega(\cdot)$. The corresponding $\sigma$-algebras of Borel sets will be denoted by $\B(\R^3)$, respectively $\B(S^2)$.

The present section summons up the theoretic results concerning density and entropy estimation of the directional distribution $f$ in this particular $3$-dimensional setting. We refer to~\cite{ARS17} for the more general case of MPPPs with manifold-valued marks.
\subsection{Kernel estimator of the directional distribution density}
Let $B\in\B(\R^3)$ denote the observation window (image, piece of material) that provides the data. For simplicity, one can think of $B=[0,b]^3$ for some $b>0$. The estimator of the directional density $f\colon S^2\to\R$ on $B$ proposed in~\cite{ARS17} is given by
\begin{equation}\label{eqn def esti f}
\hat{f}_{B}(\eta):=\frac{1}{\lambda\vol(B)}\sum_{i\geq 1}\frac{\mathds{1}_{\{Y_i\in B\}}}{h^{2}\theta_\eta(\xi_i)}K\left(\frac{d_g(\eta,\xi_i)}{h}\right),\qquad \eta\in S^2,
\end{equation} 
where $h>0$ denotes the \textit{bandwidth}, $\theta_\eta(\xi_i)=|\sin \arccos \langle \eta, \xi_i\rangle|/\arccos \langle \eta, \xi_i\rangle$ is the so-called \textit{volume density function}, $K\colon\R_+\to\R$ is a \textit{kernel function}, and $d_g(\eta,\xi_i)= \arccos\langle \eta, \xi_i\rangle$ is the \textit{geodesic distance} between $\eta,\xi_i\in S^2$. Here, $\langle\cdot,\cdot\rangle$ denotes the standard scalar product on $\R^3$.

Under the proper choice of a kernel $K$ and a sequence of bandwidths $\{h_n\}_{n\geq 1}$ (see~\cite[Corollary 3.2]{ARS17} for details), this estimator is $L^2$-consistent, i.e.
\begin{equation*}
\E[\Vert\hat{f}_{B_n}- f\Vert^2_2]\xrightarrow{n\to\infty}0,%\leq\frac{4C_\theta\pi K^2_0}{\lambda \abs{B_n}h_n^2}+\frac{3}{4}\pi h_n^4 C_2^2 K_2^2
\end{equation*}
where $\{B_n\}_{n\geq 1}\subseteq\R^3$ is a sequence of regularly growing Borel sets, for instance $B_n=[0,b_n]^3$, with $b_n\to\infty$, and $\|f\|_2^2=\int_{S^2}f^2(\xi)\,\omega(d\xi)$. In addition, if the observation windows $B_n$ are large enough, $\hat{f}_{B_n}$ is almost surely consistent as well (c.f.~\cite[Theorem 3.8]{ARS17}). Moreover, it is possible to give an asymptotically optimal bandwidth 
\[
h_{opt}=\left(\frac{C_\theta K_0^2}{2C_2^2K_2^2\lambda\vol(B)}\right)^{1/6},
\]
where the constants $C_\theta,K_0,K_2,C_2>0$ depend on the volume density function $\theta_\eta$, the kernel $K$ and the density $f$.
\subsection{Entropy estimator}
The non-parametric estimator of the directional distribution $f$ in the observation window $B$ is defined as
\begin{equation}\label{eq def esti e}
\widehat{\e}_f(B):=-\frac{1}{\lambda\vol(B)}\sum_{i\geq 1} \mathds{1}_{\{Y_i\in B\}}\log\hat{f}_{B'+Y_i}(\xi_i),
\end{equation}
where $B'\subseteq B$ is the sub-window in which the density $f$ is estimated and $B'+y:=\{x+y,~x\in B'\}$ denotes the translation of $B'$ by $y\in\R^3$. The additional window $B'$ is introduced for the purpose of notation and it is only relevant when studying the asymptotic normality of the estimator. For the purpose of estimation and consistency there are no restrictions on it and one can assume $B'=B$. 

Under the proper assumptions on the kernel $K$, the sequence of bandwidths $\{h_n\}_{n\geq 1}$ and some integrability properties of $f$ (see~\cite[Theorem 4.1]{ARS17} for details), this estimator is $L^2$-consistent, i.e.
\begin{equation*}
\E[\vert\widehat{\e}_f(B_n)-\e_f\vert^2]\xrightarrow{n\to\infty}0,
%\leq 3\left(\frac{8K_0C_\theta\upsilon_g(M)}{\lambda^2\edit{|B_n'|^2}b_n^p}+\frac{4}{\lambda^2|B_n'|}+16b_n^2L_2+\frac{L_1}{\lambda\edit{|B'_n|}}\right)
\end{equation*}
for a sequence of regularly growing Borel sets $\{B_n\}_{n\geq 1}\subseteq\R^d$. The bias of $\widehat{\e}_f(B_n)$ is controlled by the bandwidth $h_n$ for which it is possible to give the asymptotically optimal value
\begin{equation}\label{eq optimal h}
h_{opt}=\Big(\frac{2\pi K_0C_\theta}{L_2\lambda^2\vol(B)\vol(B')}\Big)^{1/4},
\end{equation}
where the constants $C_\theta,K_0,L_2>0$ depend on $f$ (c.f.~\cite[Remark 4.4]{ARS17} and~\cite[Remark 4.3.3]{Sch16}).

Due to the complex dependency structure of the random field $\{\log_{B'+Y_i}(\xi_i)\}_{i\geq 1}$, it is necessary to modify the estimator in order to study its asymptotic normality. Namely, we consider
\begin{equation}\label{eq def modified e}
\widehat{\e}^*_f(B):=-\frac{1}{\lambda\vol(B)}\sum_{i\geq 1} \mathds{1}_{\{Y^*_i\in B\}}\log\hat{f}_{B'+Y^*_i}(\xi^*_i),
\end{equation}
where $\Psi^*_\lambda=\{(Y_i^*,\xi^*_i)\}_{i\geq 1}$ is an independent copy of the original MPPP $\Psi_\lambda$. It is shown in~\cite[Theorem 5.7]{ARS17} that under suitable assumptions, in particular when the sequences of growing windows $\{B_n\}_{n\geq 1}$ and $\{B_n'\}_{n\geq 1}$ satisfy $B_n=[0,b_n]^3$, $B_n'=[0,m_n]^3$ and $b_n=m_n^{4+\delta}$ for some $\delta>0$, the central limit theorem
\begin{equation}\label{eq CLT}
\sqrt{\vol(B_n)}\frac{\widehat{\e}^*_f(B_n)-\widehat{\mu}_{B_n}}{\sigma_n}\xrightarrow{\quad d\quad}\mathcal{N}(0,1)
\end{equation}
holds. Here,
\begin{align}\label{eq mu sigma CLT}
&\widehat{\mu}_{B_n}=\frac{\#(\Pi_\lambda\cap B_n)}{\lambda\vol(B_n)}\E[-\log\hat{f}_{B'_n}(\xi^*_0)],\\
&\sigma^2_n=\Var(\log\hat{f}_{B'_n}(\xi^*_0))+\lambda^2\int_{B_n'}\Cov(\log\hat{f}_{B'_n}(\xi^*_0),\log\hat{f}_{B'_n}(\xi'_y))\,dy,\nonumber
\end{align}
where $\{\xi_y', y\in\R^{3}\}$ are i.i.d. copies of $\xi_0$ and $\# B$ is the cardinality of a finite set $B$. A simulation study of this limit theorem is presented in Section~\ref{section simulations}.

\medskip

We would like to remark that the results mentioned in this section also hold for MPPPs whose marks $\xi_i$ take values in more abstract spaces. Thus, the marks need not represent the direction but any other characteristic of a fibre, as long as its distribution has a density satisfying certain constraints. These are nevertheless fulfilled by distributions of many quantities of interest, such as fibre length or curvature.

\section{Detection of inhomogeneities. Methodology.}\label{section method}
How can a fibre cluster in a FRP be detected? These kind of inhomogeneities arise for instance when the production process fails to orient the fibres of a FRP properly. This section proposes an entropy-based method to answer this question that relies on the entropy of the directional distribution of the fibres in view of the ability of entropy to perceive abrupt changes in a distribution. The aforementioned estimator $\widehat{\e}^*_f(B)$ and its asymptotic properties will thus play a decisive role in this procedure. 

An important necessary condition for a successful performance will be that the directional distribution of the fibres in and outside the inhomogeneity region strongly differ from each other. Further possible limitations of our method will be discussed at the end of the section.

\subsection{Description of the method}
Let $W\subset\R_+^3$ denote the window (piece of material) to be analysed, $A\subset W$ the inhomogeneity region, and $B\subset\R_+^3$ the scanning window. For simplicity, these sets will be assumed to be cubes of side lengths $w,a,b>0$ respectively, i.e. $W=[0,w]^3$, $A=[0,a]^3+a_0$ for some $a_0\in\R^3_+$, and $B=[0,b]^3$. Moreover, we will assume that $0<b<a<w$. In particular, $b<a$ means that the volume of the scanning window is smaller than the volume of the inhomogeneity region. Notice that a scanning window that is significantly bigger than the inhomogeneity will fail to identify that region. Moreover, the simulation study has indicated this assumption to provide the best performance (see e.g. Figure~\ref{fig limitations}). 

The presence of $A$ will be detected by means of hypothesis testing. Under the hypothesis of homogeneity, the directional distribution is the same for any fibre, and we will assume that this distribution has a density $f$. Rejecting the hypothesis will be an indicator of the existence of $A$.

In a first step, the window $W$ will be scanned using the window $B$ by letting it run through the points in the minus-sampled window $W\ominus B$ to avoid boundary effects. Recall that given two sets $B_1,B_2\subset\R^3$, their \textit{Minkowski difference} $B_1\ominus B_2$ is defined as the set $\{x-y~\colon~x\in B_1, y\in B_2\}$.

At each point $x\in W\ominus B$, the entropy will be locally estimated in $B+x$, generating the stationary random field 
\[
\{\widehat{\e}_f^*(B+x), x\in W\ominus B\}.
\]

For a finite number of observation points $x_1,\ldots, x_n\in W\ominus B$, we set 
\begin{align*}
&\tilde{\mu}(\neu{\{x_1,\ldots ,x_n\}}):=\median\{\widehat{\e}_f^*(B+x_1),\ldots,\widehat{\e}_f^*(B+x_n)\},\\
&\hat{\mu}(\neu{\{x_1,\ldots ,x_n\}}):=\frac{1}{n}\sum_{i=1}^n\widehat{\e}_f^*(B+x_i),\\
&\hat{\sigma}^2(\neu{\{x_1,\ldots ,x_n\}}):=\frac{1}{n-1}\sum_{i=1}^n(\widehat{\e}_f^*(B+x_i)-\hat{\mu}(\{x_1,\ldots ,x_n\}))^2.
\end{align*}
In order to quantify a substantial deviation from the hypothesis of homogeneity, we use in a second step the asymptotic Gaussianity of $\widehat{\e}_f^*(B)$ stated in~\eqref{eq CLT} to view $\{\widehat{\e}_f^*(B+x), x\in W\ominus B\}$ approximately as a Gaussian random field (GRF). To each of its elements, the so-called \textit{$3\sigma$-rule} may be applied, which states that % the probability of a unimodal
for a Gaussian random variable $X$  it holds that%falling away from its mean by more that 3 standard deviations is at most $5\%$, i.e.
\begin{equation}\label{eq 3sigma rule}
\pr\big(|X-\E[X]|>3\sqrt{\Var(X)}\big)\approx 0.0027.
\end{equation}
This result can be extended to arbitrary distributions with unimodal density, its mode as center, and the upper bound $0.05$ instead of $0.0027$ in~\eqref{eq 3sigma rule}, see~\cite{Puk94} for a review. This rule underlies our proposed method to detect the inhomogeneity region $A$.

Recall that the direction of the fibres lying inside the region $A$ is supposed to strongly differ from the direction of the fibres outside of $A$. If the value of the local entropy $\widehat{\e}_f^*(B+x)$ substantially deviates from the ``median entropy'' $\tilde{\mu}(\{x_1,\ldots ,x_n\})$, this will indicate that the hypothesis is violated, suggesting that the point $x$ lies in $A$. The employment of the median rather than the empirical mean is our proposal to avoid outliers. The estimated inhomogeneity region will thus be defined as the excursion set
\begin{equation*}
\widehat{A}^{(n)}_{B,W}:=\{x\in W\ominus B~\colon~|\widehat{\e}_f^*(B+x)-\tilde{\mu}_n|>3\hat{\sigma}_n\},
\end{equation*}
where $\tilde{\mu}_n=\tilde{\mu}(\{x_1,\ldots ,x_n\})$ and $\hat{\sigma}_n=\hat{\sigma}^2(\{x_1,\ldots,x_n\})$ with $x_1,\ldots ,x_n\in W\ominus B$. In practice, we will take  $\tilde{\mu}=\tilde{\mu}((W\ominus B)\cap r\Z^3)$ and $\hat{\sigma}=\hat{\sigma}((W\ominus B)\cap r\Z^3)$, where $r\Z^3$ denotes the lattice of mesh size $r>0$. In simulations, these lattice points will also be used to discretize $W\ominus B$.

\subsection{Scanning window size}
A significant issue concerning the proposed method is the choice of the size of the scanning window $B$. On the one hand, $B$ should be large enough so that the local entropy is estimated accurately. On the other hand, a far too large scanning window may oversee the inhomogeneity. 

In this paragraph we give an answer to this question under the condition that some a priori information about the side length relation between the inhomogeneity and the whole observation window is known. In particular we will see that, if these lengths depend linearly on each other, then the side length of the observation window $B$ is a multiple of the side length of the inhomogeneity region $A$ and the factor will depend on the quantile $\pr(|\widehat{\e}_f^*(B)-\tilde{\mu}|>3\hat{\sigma})$. Setting this quantile to $0.05$ and if for instance the side length $w$ of the observation window and the side length $a$ of the inhomogeneity satisfy the relation $w=7a$, then $B$ will be chosen to have side length $b=0.489a$.

We use the notion of \textit{expected distance in measure} to quantify the approximation error of the estimated (random) inhomogeneity region $\widehat{A}_{B,W}$. For an abstract measurable space $(M,\nu)$ and any two measurable sets $S_1,S_2\subseteq M$, the \textit{distance in measure} between them is given by the measure of their symmetric difference $S_1\small{\triangle}S_2:=(S_1\cap S_2^c)\cup (S_1^c\cap S_2)$. The analogous concept for random sets is known as the expected distance in measure, which is defined as follows.

\begin{defn}
Let $(M,\nu)$ be a a measurable space. The expected distance in measure with respect to $\nu$ is the function $d_\nu\colon\B(M)\times\B(M)\to[0,\infty]$ defined as
\[
d_\nu(\Gamma_1,\Gamma_2):=\E[\nu(\Gamma_1\triangle\Gamma_2)]
\]
for any Borel measurable random closed sets $\Gamma_1,\Gamma_2\colon \Omega\to\B(M)$.
\end{defn}
In view of this definition, the size, i.e. the Lebesgue measure or volume, of the scanning window $B$ is optimal when the expected distance in measure between the original inhomogeneity region $A$ and the estimated region $\widehat{A}_{B,W}$ is minimal. In other words, $B$ will be chosen to satisfy
\[
B=\argmin\limits_{\tilde{B}}d_{\vol}(A,\widehat{A}_{\tilde{B},W}). 
\]
Notice that the MPPP $\Psi_{\lambda,A}:=\{(Y_i,\xi_i)\}_{i\geq 1}$ modeling a fibre system with an inhomogeneity has marks $\xi_i$ that are independent but \textit{not} identically distributed. Namely, marks that correspond to points in $A$ follow a directional distribution with density $g$ whereas marks corresponding to points in $W\setminus A$ have a different direction with distribution density $f$. In this case, the entropy estimator will be simply denoted by $\widehat{\e}^*$ since there is no specific density distribution to refer to.
%In this situation, the entropy estimator becomes
%
%\begin{equation}\label{eq def esti e2densities}
%\widehat{\e}^*(B):=\frac{1}{\lambda\vol(B)}\sum_{i\geq 1}\Big(\mathds{1}_{\{Y^*_i\in B'\cap A\}}\log\hat{g}_{B'+Y^*_i}(\xi^*_i)+\mathds{1}_{\{Y^*_i\in B\cap A^c\}}\log\hat{f}_{B'+Y^*_i}(\xi^*_i)\Big)
%\end{equation} 
%with $B'\subseteq B\subseteq\R^3$ is defined following~\eqref{eq def esti e}. This estimator is written without any specific subscript because there are two different directional densities on play. 

For simplicity of the calculations and to avoid further boundary effects, we will assume that the observation window $W$ is large enough and the inhomogeneity region $A$ small enough so that 
\begin{equation}\label{eq size condition A and W}
A\oplus B\subseteq W\ominus B.
\end{equation}
The expected distance in measure with respect to the measure $\vol(\cdot)$ has the following expression.

\begin{lemma}\label{lem dvol}
Let $R_1=( A\oplus B)\setminus A$ and let $R_2=A\setminus (A\ominus B)$, where $\partial A$ denotes the boundary of $A$. Then,
\begin{align*}
d_{\vol}(A,\widehat{A}_{B,W})&=\vol(A)+\vol\big((W\ominus B)\setminus(A\oplus B)\big)\pr(|\widehat{\e}_f^*(B)-\tilde{\mu}|>3\hat{\sigma})\\
&-\vol(A\ominus B)\pr(|\widehat{\e}_g^*(B)-\tilde{\mu}|>3\hat{\sigma})\\
&+\int_{R_1}\pr(|\widehat{\e}^*(B+x)-\tilde{\mu}|>3\hat{\sigma})dx-\int_{R_2 }\pr(|\widehat{\e}^*(B+x)-\tilde{\mu}|>3\hat{\sigma})dx.
\end{align*}
\end{lemma}
\begin{proof}
On the one hand, by definition of expected distance and symmetric difference we have that %($\vol(C\triangle D)=\vol(C)+\vol(D)-2\vol(C\cap D)$
\begin{equation}\label{eq dvol}
d_{\vol}(A,\widehat{A}_{B,W})=\vol(A)+\E[\vol(\widehat{A}_{B,W})]-2\E[\vol(A\cap\widehat{A}_{B,W})].
\end{equation}
On the other hand, Fubini's theorem yields
\[
\E[\vol(\widehat{A}_{B,W})]=\int_{W\ominus B}\pr(|\widehat{\e}^*(B+x)-\widetilde{\mu}|>3\hat{\sigma})\,dx
\]
and since %$A\cap\widehat{A}_{B,W}=\widehat{A}_{B,A}$
$A\subseteq W\ominus B$ by~\eqref{eq size condition A and W}, this also yields
\[
\E[\vol(A\cap\widehat{A}_{B,W})]=\int_A\pr(|\widehat{\e}^*(B+x)-\widetilde{\mu}|>3\hat{\sigma})\,dx.
\]
Let us analyze the first integral. Notice that if $x\in (W\ominus B)\setminus(A\oplus B)$, then $B+x$ does not intersect $A$ and therefore $\widehat{\e}^*(B+x)=\widehat{\e}^*_f(B+x)$. On the other hand, if $x\in A\ominus B$, then $\widehat{\e}^*(B+x)=\widehat{\e}^*_g(B+x)$. In view of the stationarity of $\{\widehat{\e}^*_f(B+x),\,x\in \R^3\}$ and $\{\widehat{\e}^*_g(B+x),\,x\in \R^3\}$ we have that
\begin{align*}
&\int_{W\ominus B}\pr(|\widehat{\e}^*(B+x)-\widetilde{\mu}|>3\hat{\sigma})\,dx\\
&=\vol\big((W\ominus B)\setminus(A\oplus B)\big)\pr(|\widehat{\e}_f^*(B)-\widetilde{\mu}|>3\hat{\sigma})\\
&+\vol(A\ominus B)\pr(|\widehat{\e}_g^*(B)-\widetilde{\mu}|>3\hat{\sigma})+\int_{R_1\cup R_2}\pr(|\widehat{\e}^*(B+x)-\widetilde{\mu}|>3\hat{\sigma})\,dx.
\end{align*}
Analogously we obtain
\begin{align*}
&\int_A\pr(|\widehat{\e}^*(B+x)-\widetilde{\mu}|>3\hat{\sigma})\,dx\\
&=\vol(A\ominus B)\pr(|\widehat{\e}_g^*(B)-\widetilde{\mu}|>3\hat{\sigma})+\int_{R_2}\pr(|\widehat{\e}^*(B)-\widetilde{\mu}|>3\hat{\sigma})\,dx.
\end{align*}
Plugging these equalities into~\eqref{eq dvol} leads to the desired expression.
\end{proof}
This result can be applied in order to obtain an upper bound for the distance $d_{\vol}(A,\widehat{A}_{B,W})$ in terms of the side lengths of $W,A$ and $B$. The side length of the scanning window $B$ will be chosen to minimize this bound and it will depend on the side lengths of $W$ and $A$.

Taking into account assumption~\eqref{eq size condition A and W}, the inhomogeneity region can be expressed as $A=[0,a]^3+a_0$ for some $a>0$ and $a_0\in W\setminus B''$ with $B''=[0,2b+a]^3$. In this case, the following identities hold
\[
W\ominus B=[0,w-b)^3,\quad A\ominus B=[0,a-b)^3+a_0,\quad A\oplus B=[0,a+b)^3+a_0.
\]
Further, set 
\[
\alpha_f:=\pr(|\widehat{\e}_f^*(B)-\tilde{\mu}|>3\hat{\sigma}),\quad\text{and}\quad\alpha_g:=\pr(|\widehat{\e}_g^*(B)-\tilde{\mu}|>3\hat{\sigma}).
\]
In view of Lemma~\ref{lem dvol} and since $\vol(R_1)=\vol(A\oplus B)-\vol(A)$,
\begin{align*}
d_{\vol}(A,\widehat{A}_{B,W})&\leq\vol(A)+\vol\big( (W\ominus B)\setminus (A\oplus B)\big)\alpha_f-\vol(A\ominus B)\alpha_g+\vol(R_1)\\
&=\vol(A\oplus B)(1-\alpha_f)+\vol(W\ominus B)\alpha_f-\vol(A\ominus B)\alpha_g\\
&=(a+b)^3(1-\alpha_f)+(w-b)^3\alpha_f-(a-b)^3\alpha_g.
\end{align*}
Furthermore, the entropy of the directional distribution of fibres lying in $A$ must deviate from the median $\widetilde{\mu}$ with higher probability than when the fibres lie outside of $A$. Hence, we can assume that $\alpha_f<\alpha_g$ and therefore get
\[
d_{\vol}(A,\widehat{A}_{B,W})\leq (a+b)^3(1-\alpha_f)+\big((w-b)^3-(a-b)^3\big)\alpha_f.
\]
With the help of computational software (here Mathematica was used) the value of %(\textbf{use mathematica or any other program})
\[
\argmin_b\{(a+b)^3(1-\alpha_f)+\big((w-b)^3-(a-b)^3\big)\alpha_f\}
\]
is given by the expression
\begin{equation}\label{eq optimal scanning window}
b_{opt}=\frac{\sqrt{\alpha_f(w-a)(w+(3-4\alpha_f)a)}-(1-2\alpha_f)a-\alpha_f w}{1-\alpha_f}.
\end{equation}
Since $a<w$ and $\alpha_f$ small, it holds that $\alpha_f(w-a)(w+(3-4\alpha_f)a)>0$. Yet, the parameter $\alpha_f$ and the relation between $a$ and $w$ must be chosen in order to assure $0<b_{opt}<a$. %$\alpha_f$ should be relatively small (actually less than 0.05 in view of~\eqref{eq 3sigma rule}). In case that $\alpha_g<\alpha_f$, then $\alpha_g$ is at least that small as well. 
As mentioned at the beginning of this paragraph and following the $3\sigma$-rule, setting $\alpha_f=0.05$ we chose $w=7a$ for the simulation study, for which~\eqref{eq optimal scanning window} yields $b=0.489a$. 

\section{Simulations}\label{section simulations}
In order to validate the estimators introduced in Section~\ref{section theory} as well as to examine the efficiency of the inhomogeneity detection method presented Section~\ref{section method}, this section is devoted to several tests on simulated data. One of the main reasons for working with simulations is that the input is known and therefore the estimation error can be computed accurately. Moreover, the MPPPs considered for these simulations model structures with a high fibre intensity, for which the analysis of real $\mu$CT data is usually difficult and costly. The high intensity plays an important role in the estimation of the entropy and therefore in the detection method: since the estimators $\hat{f}_B$ and $\widehat{\e}^*$ are asymptotically consistent, the amount of data (fibres) has to be sufficiently large in order to provide a good approximation.

\subsection{Density estimation}
In order to test the density estimator $\hat{f}_B$ defined in~\eqref{eqn def esti f}, an independently marked MPPP of intensity $15$ was simulated in an observation window $B=[0,50]^3$. Table~\ref{Kernels table} and Figure~\ref{Kernel plots} display the approximation error of $\hat{f}_B$ for different kernels (see e.g.~\cite{Tsy09} for definitions and further details about kernel functions). Here, the computed error is given by $\max_{x\in S^2_d}|\hat{f}_B(x)-f(x)|$, where $S^2_d$ denotes a discretization of the sphere into 
%how many??
a grid of points. Since the constants appearing in the expression~\eqref{eq optimal h} of the optimal bandwidth cannot be computed explicitly, by empirical observation the bandwidth was set to be $h=\big(\frac{1 + \vol(B)}{\vol(B)^{10/9}}\big)^{1/4}$. Notice that this quantity has the same asymptotic order as~\eqref{eq optimal h} when $\vol(B')=\vol(B)^{1/9}$.
 
\begin{table}[H]
\centering
\begin{tabular}{l|c|c|c|c|}
 &\hspace*{.1in}Uniform\hspace*{.1in}&\hspace*{.05in}$\operatorname{Schladitz}(2)$\hspace*{.05in}& $\operatorname{Watson}(2)$ & $\operatorname{Fisher}(2)$\\ \hline 
Biweight 		& $0.0094$ & $0.0125$ & $0.2487$ & $0.0152$\\ %\hline
Epanechiknov 	& $0.0195$ & $0.0218$ & $0.3393$ & $0.0499$\\ %\hline
Triangular		& $0.0089$ & $0.0089$ & $0.2751$ & $0.0226$\\
Tricube			& $0.0082$ & $0.0098$ & $0.2536$ & $0.0149$\\
Triweight		& $0.0190$ & $0.0285$ & $0.2487$ & $0.0521$\\
Uniform			& $0.0524$ & $0.0654$ & $0.5284$ & $0.1565$
\end{tabular}
\caption{Approximation error of different distribution densities $f$ by its kernel density estimator $\hat{f}_B$ for different kernels measured in the uniform convergence metric.}
\label{Kernels table}
\end{table}

The density $f$ corresponds in each case to the uniform, Schladitz(2), Watson(2) or Fisher(3) directional distribution on the sphere, see e.g.~\cite{FLE93,S+06} for precise definitions of these distributions. 

\bigskip

\begin{figure}[H]
\centering
\begin{tabular}{cccc}
\includegraphics[scale=0.4]{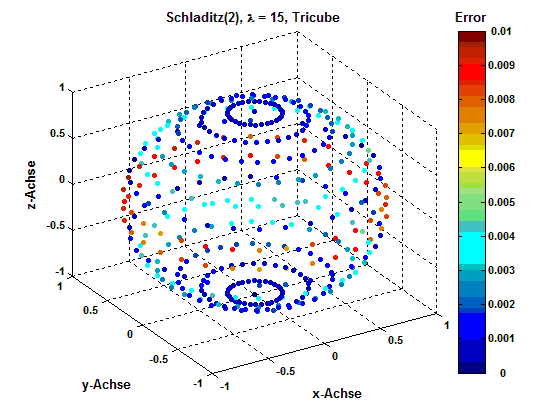}
&
\includegraphics[scale=0.4]{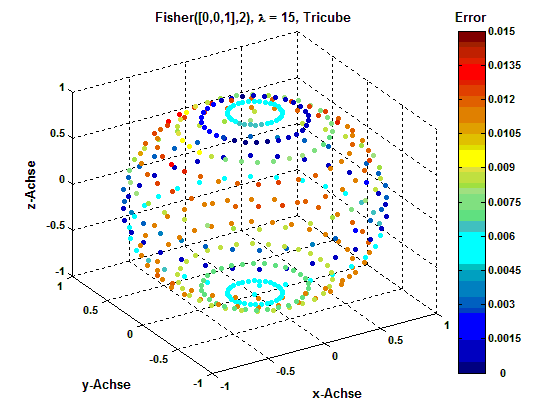}
\end{tabular}
\caption{\footnotesize{ Graphical visualization of the estimation error for the directional distributions $\operatorname{Schladitz}(2)$ and $\operatorname{Fisher}([0,0,1],2)$ with Tricube kernel.}}
\label{Kernel plots}
\end{figure}
An exhaustive simulation study with more distributions, kernels and observation window sizes can be found in~\cite{Sch16}.
\subsection{Entropy estimator}
For the study of the entropy estimator defined in~\eqref{eq def esti e}, an independently marked MPPP of intensity $15$ was simulated in an observation window $B=[0,50]^3$. Recall that the density $f$ is estimated in the additional window $B'=[0,4]^3$. Although $B'$ could still have been chosen to equal the whole observation window $B$ at this point, we took on account as far as possible the side-length relation between these windows appearing in the asymptotic normality test based on~\eqref{eq CLT}. The window $B'$ being small, it has to contain sufficiently many points (fibres) to obtain a fairly good estimation. This is guaranteed by the fact that the intensity of the underlying point process is relatively high.

In the subsequent testing, the kernel $K$ is fixed to tricube because the previous simulations indicated it to yield the best approximation. The bandwidth is the same as for the density estimation, $h=\big(\frac{1 + \vol(B)}{\vol(B)^{10/9}}\big)^{1/4}$, with the same asymptotic order as~\eqref{eq optimal h} and $\vol(B')=\vol(B)^{1/9}$.

\begin{table}[H]
\centering
\begin{tabular}{l|c|c|c|c|c|c|}
& $\e_f$ & $\;\overline{\widehat{\e}_f(B)}\;$ & $\;\operatorname{Var}\widehat{\e}_f(B)\;$ &$\Vert Err\Vert_\infty$ &  MSQE\\ \hline 
Uniform 		& 2.5310 & 2.5165 & 8.7105e-07 & 0.0157 & 0.0459\\ %\hline
Schladitz(2) 	& 2.3554  & 2.3525 & 9.0843e-07 & 0.0039 & 0.0067\\ %\hline
Fisher(2) & 1.7239 & 1.8930 & 2.3662e-06 & 0.1715 & 0.3782\\
Watson(2)	& 1.8646 & 1.6849 & 2.8397e-07 & 0.1804 & 0.5682\\
\end{tabular}
\caption{True value, arithmetic mean, sample variance, absolute error and mean square error of the entropy estimation.}
\end{table}

\begin{figure}
\centering
\begin{tabular}{ccc}
\includegraphics[scale=0.23]{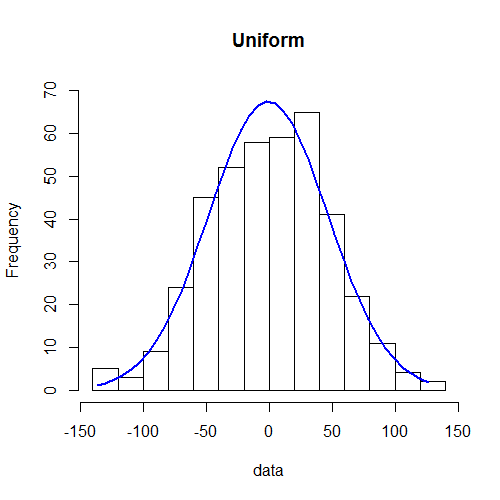}&
\includegraphics[scale=0.23]{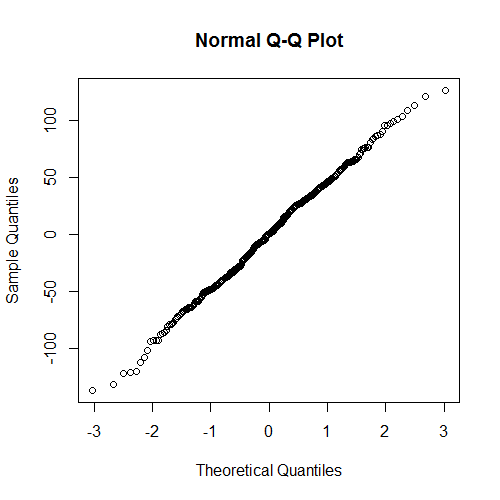}&
\includegraphics[scale=0.23]{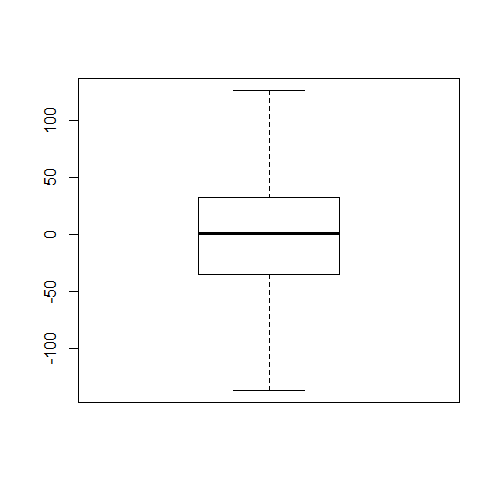}\\
\includegraphics[scale=0.23]{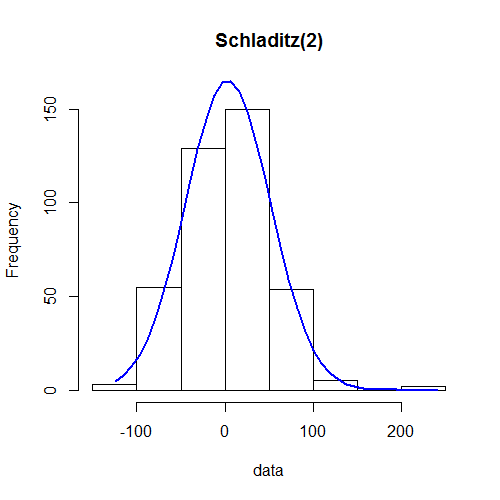}&
\includegraphics[scale=0.23]{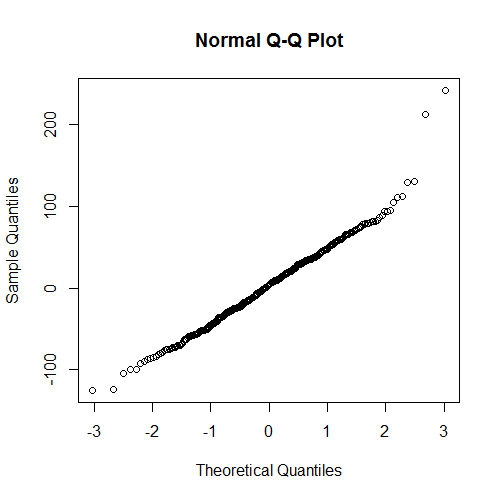}&
\includegraphics[scale=0.23]{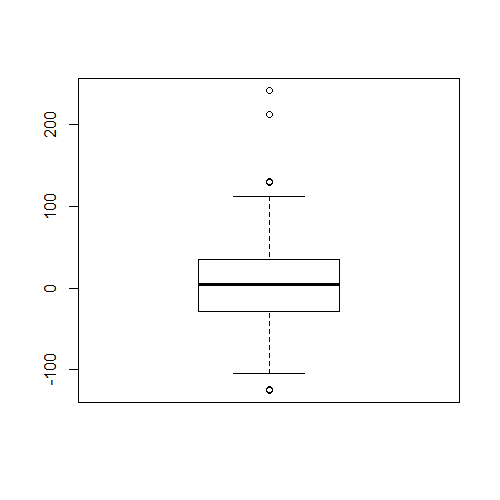}\\
\includegraphics[scale=0.23]{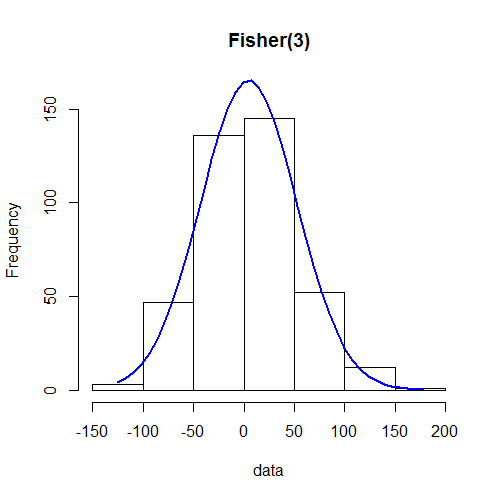}&
\includegraphics[scale=0.23]{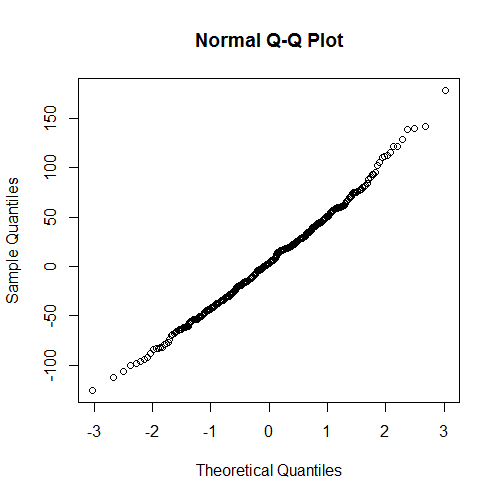}&
\includegraphics[scale=0.23]{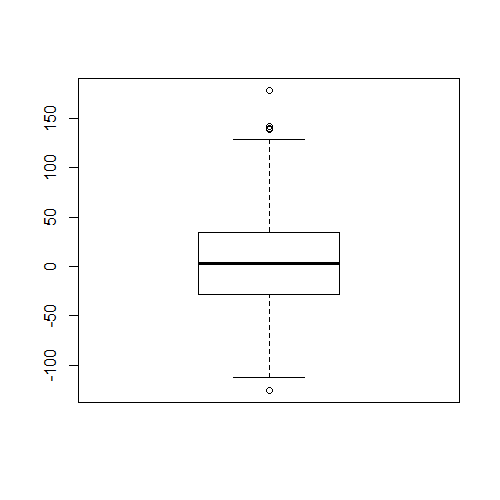}
\end{tabular}
\caption{Fitted histogram, q-q plot and box plot for the entropy estimator $\e^*_f(B)$ with uniform distribution, $\operatorname{Schladitz}(2)$- and $\operatorname{Fisher}(3)$-distribution, respectively.}
\label{CLT study}
\end{figure}

The asymptotic normality of the modified estimator $\widehat{\e}_f^*(B)$ introduced in~\eqref{eq def modified e} has been tested by running 400 simulations of $\widehat{\e}_f^*(B)$ based on two independent copies, $\Psi_{20}$ and $\Psi^*_{20}$, of an independently marked MPPPs of intensity $20$  observed in a window $B=[0,30]^3$. The auxiliary window where the density is estimated is $B'=[0,4]^3$. Figure~\ref{CLT study} summarizes the results for different directional distributions.

In the computation of the normalization terms in~\eqref{eq mu sigma CLT}, the mean has been replaced by its empirical estimation obtained from $180$ independent realizations of $-\log\hat{f}_{B'}(\xi^*_0)$. The integral in the covariance term has been discretized using $343$ latice points in $B'$ and the corresponding empirical estimator of each covariance term. %$q_n^d=\vol(B')^{1+\frac{1}{13}}$

The size of the observation windows and the high density of the underlying point process reproduce the asymptotic nature of the theoretic result. %Notice that the expected number of points in the process lies by $20\cdot30^3=180000$. Nevertheless, 
\subsection{Inhomogeneity detection}
In this section we present two examples where our proposed method detects a (self-generated) inhomogeneity region. The size relationship between the observation window, the inhomogeneity region and the scanning window have been determined following the optimal expression in~\eqref{eq optimal scanning window} with $\alpha_f=0.05$ and $w=7a$. 

\begin{figure}[H]
\centering
\begin{tabular}{c c}
\includegraphics[scale=.45]{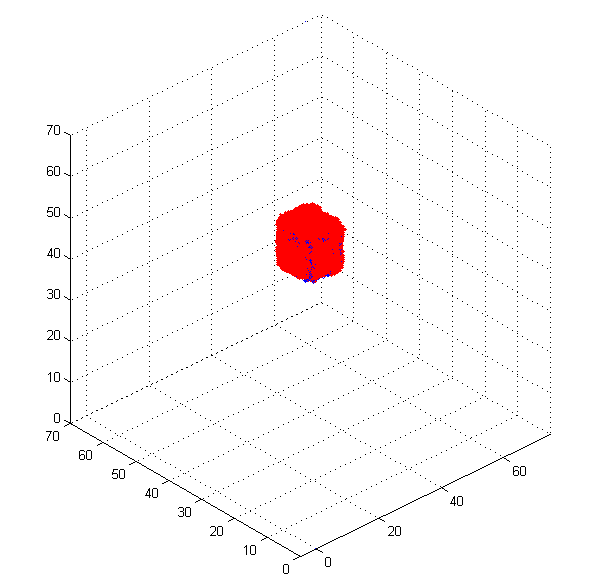}
&
\includegraphics[scale=.40]{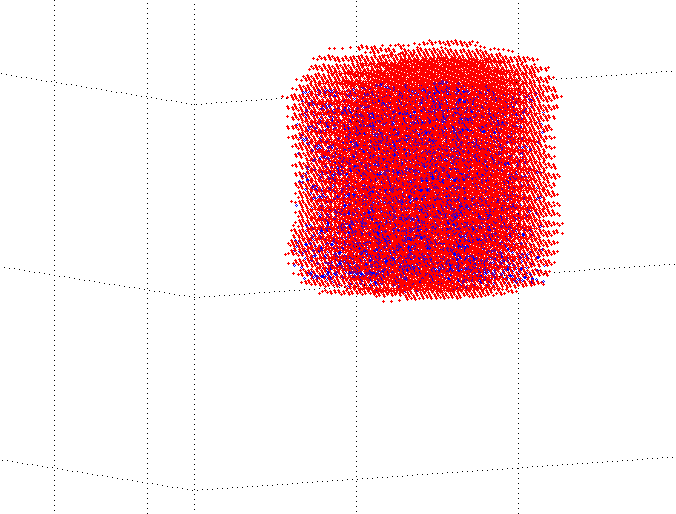}
\end{tabular}
\caption{The cube where the inhomogeneity resides is completely recognized.}
\label{detect example 1}
\end{figure}

In the first example, an underlying independently marked MPPP $\Psi_{5}$ of intensity $5$ is generated in the window $W=[0,70]^3$. The inhomogeneity region, colored in blue in Figure~\ref{detect example 1}, is the cube $A=[0,10^3]+(35,35,35)$. $W$ is scanned by running the window $B=[0,5]^3$ over $130^3$ lattice points. The fibres (points) inside $A$ have directions following a uniform distribution, whereas the direction of fibres outside the inhomogeneity region is Fisher(10)-distributed. The red points constitute the (discretized) estimated inhomogeneity region $\widehat{A}_{B,W}$ obtained as described in Section~\ref{section method}.

The second example illustrates that this method also works when there are more than just one ihnomogeneity region, as long as each component satisfies the constraints~\eqref{eq size condition A and W} and~\eqref{eq optimal scanning window}. In this case, $A=A_1\cup A_2$ with $A_1=[0,10]+(10,10,10)$ and $A_2=[0,10]+(30,30,30)$. The direction of the fibres inside the inhomogeneity region $A$ follow a uniform distribution, whereas the direction of fibres outside $A$ have Fisher(10)-distributed directions. 
\begin{figure}[H]
\centering
\begin{tabular}{c c}
\includegraphics[scale=.45]{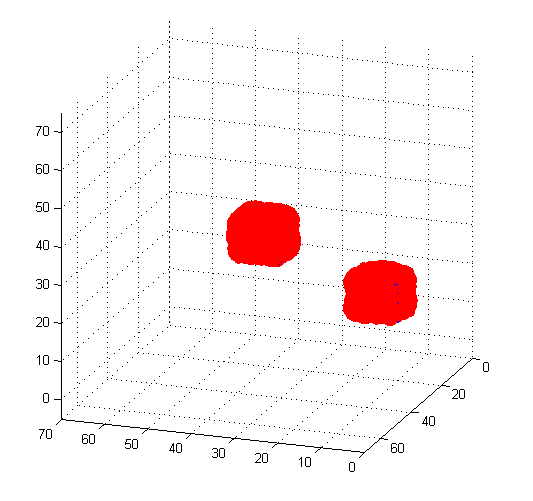}
&
\includegraphics[scale=.36]{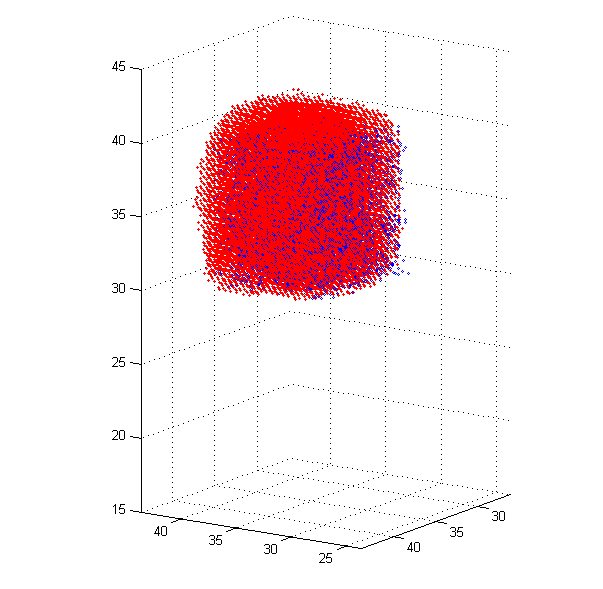}
\end{tabular}
\caption{Two cubes with marks differently distributed are also completely recognized.}
\end{figure}
%For this simulations a lattice of $100^2$ points was used. Due to the higher density of this lattice, the discretized estimated inhomogeneity region (in red) almost covers $A$.

\section{Discussion and outlook}\label{section discussion}

Inhomogeneities such as fibre clusters or deformations in FRPs are directly related to changes in the direction of the fibres that compose the material and detecting these regions is of special interest in order to control its production process.

In this paper we propose a method to detect inhomogeneities based on the entropy as a measure of change in the directional distribution of the fibres. One major advantage of using the entropy is that the problem reduces to a univariate change-point detection that can be performed via the $3\sigma$-rule due to the asymptotic normality of the entropy estimator. 

The tests presented in Section~\ref{section simulations} support the applicability of the theoretical estimator for the density and entropy of directional distributions that underlie the method exposed in Section~\ref{section method}. Within the constraints concerning the shape of the observation windows $W$, $B$ and the inhomogeneity region $A$, as well as assuming prior knowledge about the size relation beteween $W$ and $A$, it is possible to derive an optimal size for the scanning window $B$. 

A finer analysis of the output reveals some discrepancies along the boundary of the inhomogeneity $A$ as Figure~\ref{bdry effect} shows. The reason for this is that, whenever the lower-left corner of the scanning window $B$ is near the boundary $\partial A$ but still outside of $A$, and $B\cap A$ covers most of $B$, the entropy on $B$ is almost the same as when $B$ completely lies in $A$. Thus, even if the whole scanning window does not lie inside $A$, the estimated entropy $\widehat{\e}^*(B)$ will be considered an outlier.

\begin{figure}[H]
\centering
\begin{tabular}{c c}
\includegraphics[scale=.35]{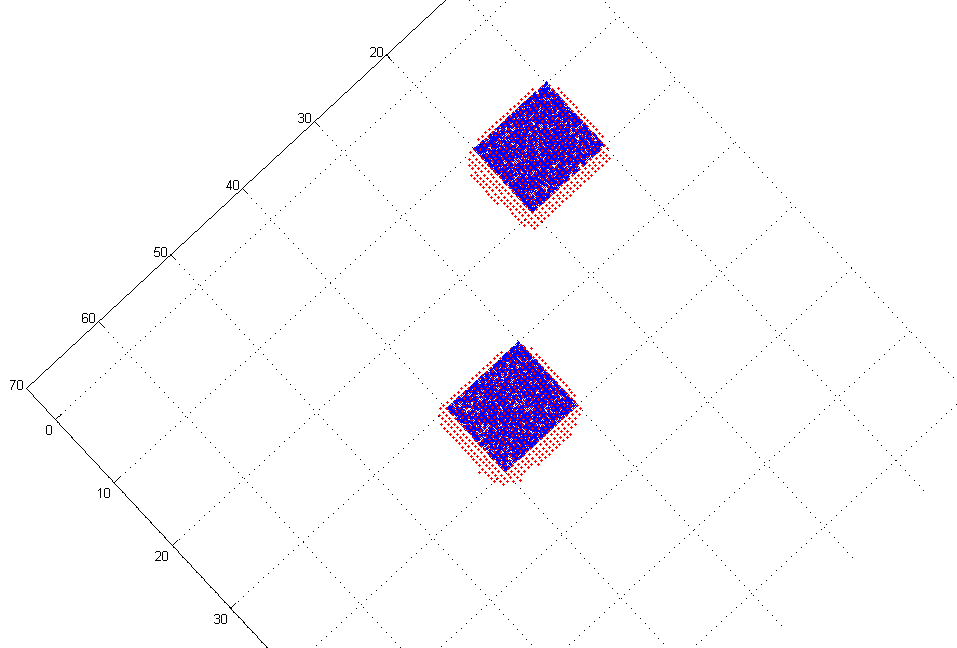}
&
\end{tabular}
\caption{Boundary effect in the inhomogeneity detection.}
\label{bdry effect}
\end{figure}

The shape of the inhomogeneity $A$ seems to affect the performance as well. On the one hand, it is not anymore possible to establish a simple relation between $A$ and the scanning window $B$ (as long as it has a shape of different nature). On the other hand, stronger boundary effects should be expected. In Figure~\ref{fig ball vs cube} the inhomogeneity has been chosen to be a ball of radius $5$ centered in the observation window, and the scanning window is been kept a cube of side-length $5$. We observe how in this case the inhomogeneity is still well recognized, however additional points appear, and the boundary effects observed in the previous example increase.

\begin{figure}[H]
\centering
\begin{tabular}{c c}
\includegraphics[scale=.30]{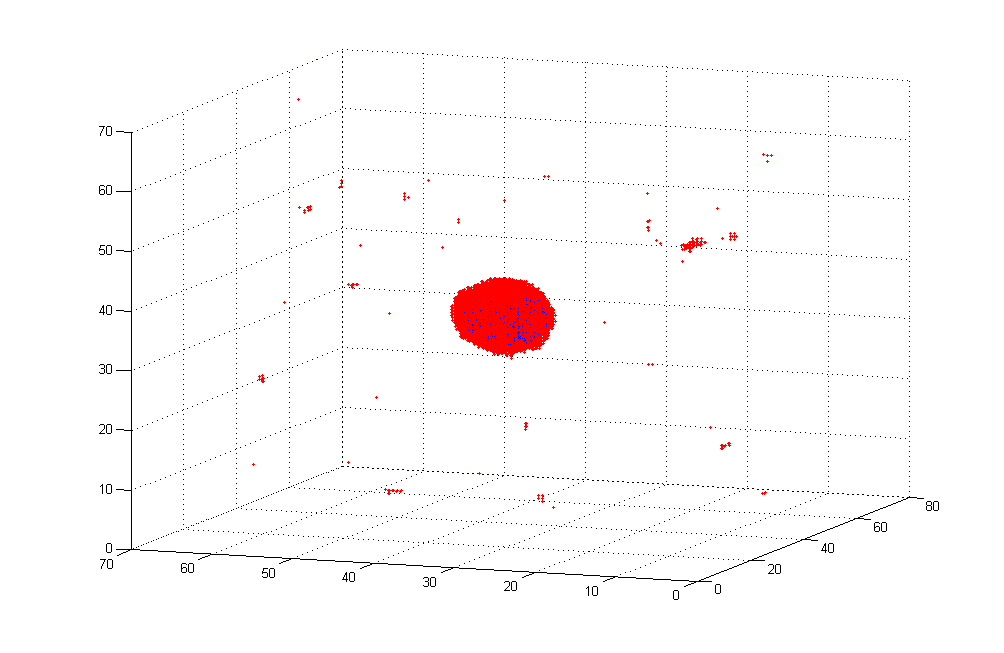}
&
\includegraphics[scale=.25]{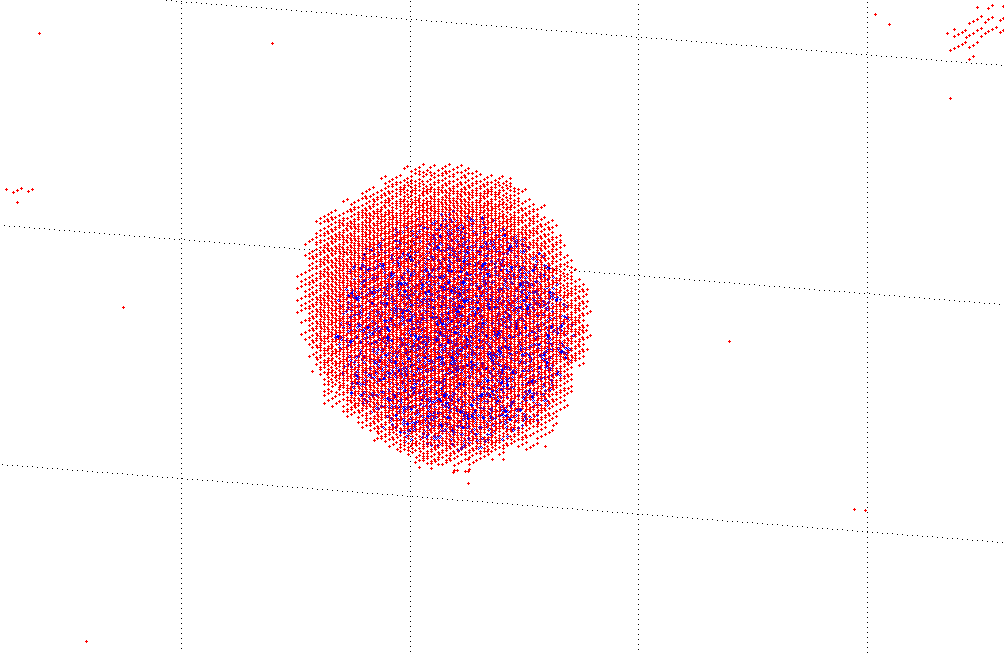}
\end{tabular}
\caption{The method becomes less efficient if the form of the inhomogeneity and of the scanning window notably differ. Boundary effects appear as well.}
\label{fig ball vs cube}
\end{figure}

Further issues appear when the size of the scanning window $B$ is much smaller or larger than the size of the inhomogeneity $A$, c.f. Figure~\ref{fig limitations}.

\begin{figure}[H]
\centering
\begin{tabular}{c c}
\includegraphics[scale=.30]{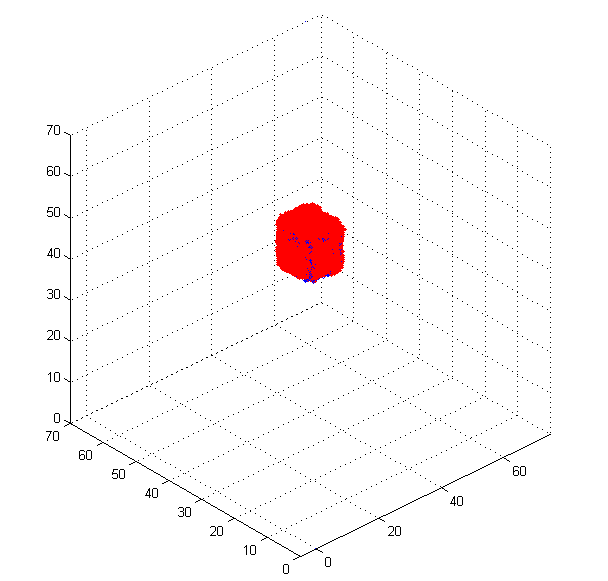}
&
\includegraphics[scale=.30]{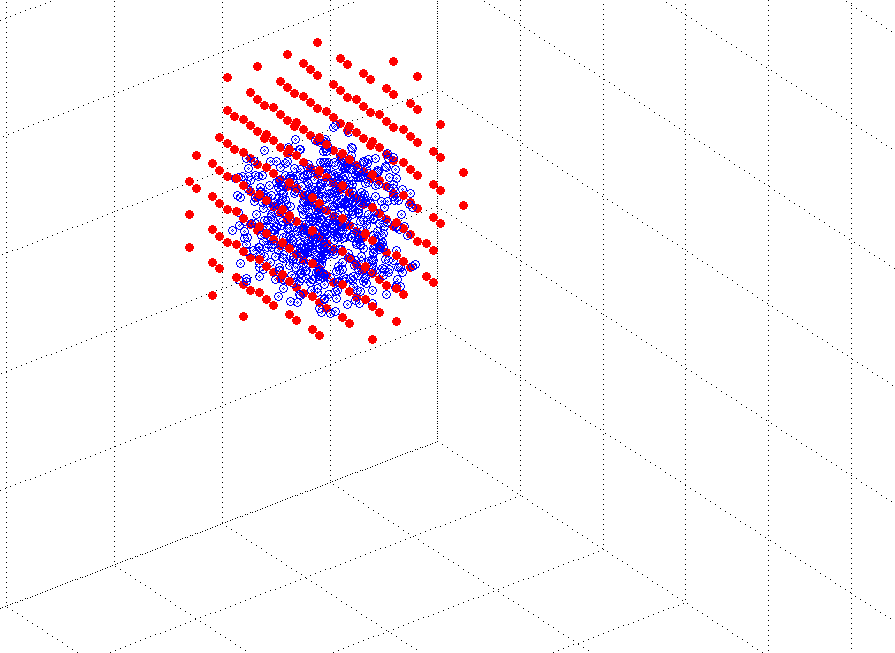}
\end{tabular}
\caption{The performance also decays when the inhomogeneity region is actually larger than expected (left) or smaller (right).}
\label{fig limitations}
\end{figure}

In conclusion, these investigations reveal our method to be a good detection procedure that may potentially be used for a first fast rough scan to rule out substantial inhomogeneities, where their contours need not be perfectly recognized.
\bibliographystyle{amsplain}
\bibliography{validation}

\providecommand{\bysame}{\leavevmode\hbox to3em{\hrulefill}\thinspace}
\providecommand{\MR}{\relax\ifhmode\unskip\space\fi MR }
\providecommand{\MRhref}[2]{%
  \href{http://www.ams.org/mathscinet-getitem?mr=#1}{#2}
}
\providecommand{\href}[2]{#2}
\begin{thebibliography}{10}

\bibitem{ARS17}
P.~Alonso-Ruiz and E.~Spodarev, \emph{Estimation of entropy for {P}oisson
  marked point processes}, Advances in {A}pplied {P}robability \textbf{49}
  (2017), no.~1, to appear.

\bibitem{FLE93}
N.~I. Fisher, T.~Lewis, and B.~J.~J. Embleton, \emph{Statistical {A}nalysis of
  {S}pherical {D}ata}, Cambridge University Press, Cambridge, 1993, Revised
  reprint of the 1987 original.

\bibitem{LZJ13}
J.~Li, X.~Zhang, and D.~R. Jeske, \emph{Nonparametric multivariate {CUSUM}
  control charts for location and scale changes}, J. Nonparametr. Stat.
  \textbf{25} (2013), no.~1, 1--20.

\bibitem{MJ93}
I.~B. MacNeill and V.~K. Jandhyala, \emph{Change-point methods for spatial
  data}, Multivariate environmental statistics, North-Holland Ser. Statist.
  Probab., vol.~6, North-Holland, Amsterdam, 1993, pp.~288--306.

\bibitem{OS09}
J.~Ohser and K.~Schladitz, \emph{3{D} {I}mages of {M}aterials {S}tructures:
  {P}rocessing and {A}nalysis}, Wiley, Weinheim, 2009.

\bibitem{Puk94}
F.~Pukelsheim, \emph{The three sigma rule}, Amer. Statist. \textbf{48} (1994),
  no.~2, 88--91.

\bibitem{R+12}
C.~Redenbach, A.~Rack, K.~Schladitz, O.~Wirjadi, and M.~Godehardt, \emph{Beyond
  imaging: on the quantitative analysis of tomographic volume data}, Int. J. of
  Materials Research \textbf{103} (2012), no.~2, 217 -- 227.

\bibitem{RSVW14}
C.~Redenbach, K.~Schladitz, I.~Vecchio, and O.~Wirjadi, \emph{Image analysis
  for microstructures based on stochastic models}, GAMM-Mitteilungen
  \textbf{37} (2014), no.~2, 281--305.

\bibitem{S+06}
K.~Schladitz, S.~Peters, D.~Reinel-Bitzer, A.~Wiegmann, and J.~Ohser,
  \emph{Design of acoustic trim based on geometric modeling and flow simulation
  for non-woven}, Computational Materials Science \textbf{38} (2006), no.~1, 56
  -- 66.

\bibitem{Sch16}
J.~Schwarz, \emph{Sch\"atzer f\"ur {D}ichte und {E}ntropie der
  {M}arkenverteilung eines markierten {P}oisson-{P}unktprozesses auf der
  {S}ph\"are}, Bachelor thesis, Ulm University, 2016.

\bibitem{Sha48}
C.~E. Shannon, \emph{A mathematical theory of communication}, Bell System Tech.
  J. \textbf{27} (1948), 379--423, 623--656.

\bibitem{SJT12}
L.~Shu, W.~Jiang, and K.~Tsui, \emph{A standardized scan statistic for
  detecting spatial clusters with estimated parameters}, Naval Res. Logist.
  \textbf{59} (2012), no.~6, 397--410.

\bibitem{Tsy09}
A.~B. Tsybakov, \emph{Introduction to {N}onparametric {E}stimation}, Springer
  Series in Statistics, Springer, New York, 2009.

\end{thebibliography}
\end{document}